\documentclass[aps,prl,reprint,superscriptaddress,nofootinbib]{revtex4-2}

\usepackage{amsmath}
\usepackage{amssymb}
\usepackage{bm}
\usepackage{physics}
\usepackage{xcolor}

\newtheorem{theorem}{Theorem}
\newtheorem{corollary}{Corollary}

\renewcommand{\Tr}{\operatorname{Tr}}
\newcommand{\G}{\Gamma}
\newcommand{\M}{\mathfrak M}
\newcommand{\Fr}{\mathcal F_R^{\mathfrak M}}
\newcommand{\PiR}{\Pi_R^{\mathfrak M}}
\newcommand{\DeltaW}{\Delta_{\rm WSU}}
\newcommand{\Dinst}{\Delta_{\rm inst}}
\newcommand{\Drec}{\Delta_{\rm rec}}

\begin{document}

\title{Residual Fisher Information in Measurement-Error Uncertainty}

\author{Masayuki Ohzeki}
\email{mohzeki@tohoku.ac.jp}
\affiliation{Department of Physics, Institute of Science Tokyo, Tokyo 152-8551, Japan}
\affiliation{Graduate School of Information Sciences, Tohoku University, Sendai 980-8579, Japan}
\affiliation{Research and Education Institute for Semiconductors and Informatics, Kumamoto University, Kumamoto 860-8555, Japan}
\affiliation{Sigma-i Co., Ltd., Tokyo 108-0075, Japan}
\author{Taiga Suzuki}
\affiliation{Department of Physics, Institute of Science Tokyo, Tokyo 152-8551, Japan}

\date{\today}

\begin{abstract}
Measurement error in the Watanabe-Sagawa-Ueda formulation is the excess estimation error of a classical record over the quantum statistical limit.
We show that, at the quantum-instrument level, this excess has an exact decomposition.
The record-only error matrix equals a residual-aware error matrix plus a positive Schur-complement term fixed by Fisher information left in the unread post-measurement system.
Thus measurement error separates into intrinsic instrument loss and record-restriction loss.
\end{abstract}

\maketitle

{\it Introduction.---}
Uncertainty relations for measurement errors depend on the operational interface assigned to a measurement.
Robertson's inequality follows from Gram-matrix positivity, and Schr\"odinger's refinement keeps the covariance that tilts the uncertainty ellipse \cite{Robertson1929,Schrodinger1930}.
The familiar product form is obtained only after projecting this matrix inequality onto chosen observable axes.
Ozawa's universally valid relation resolved the failure of the naive error-disturbance product by introducing noise and disturbance operators \cite{Ozawa2003}.
This definition is universal, but its operational interpretation as an estimation error is not straightforward.
Watanabe, Sagawa, and Ueda (WSU) instead define measurement error by how well an observable expectation can be inferred from the classical data \cite{WatanabeSagawaUeda2011,Watanabe2014,Nogami2025}.
The Cram\'er-Rao bound then converts this operational definition into inverse Fisher-information matrices; the uncertainty relation follows from comparing those matrices.
We make one further interface explicit.
A POVM gives a record $y$, whereas an instrument realizing that POVM gives $y$ and a conditional residual quantum state $R$.
A classical controller, syndrome decoder, or thermodynamic observer sees only $y$, while a quantum-aware sensor or feedback device may retain $R$.
The record-restriction gap is a positive Schur-complement correction for information retained in $R$.

{\it Estimation-theoretic measurement error.---}
We first recall the Watanabe-Sagawa-Ueda construction in the notation used below.
Let $\{\rho_\theta\}_{\theta\in\Theta}$ be a smooth family of quantum states with coordinates $\theta=(\theta^1,\ldots,\theta^d)$.
A record-only measurement produces outcomes $y$ with probability $p_\theta(y)$.
The classical score is $s_i(y)=\partial_i\log p_\theta(y)$, where $\partial_i:=\partial/\partial\theta^i$, and the record Fisher matrix is
\begin{equation}
    (J_Y)_{ij}
    =
    \sum_y p_\theta(y)s_i(y)s_j(y).
    \label{eq:record-fi}
\end{equation}
For each parameter direction, the symmetric logarithmic derivative (SLD) $L_i^Q$ is the Hermitian operator determined by
\begin{equation}
    \partial_i\rho_\theta
    =
    \frac{1}{2}
    \left(
    L_i^Q\rho_\theta
    +
    \rho_\theta L_i^Q
    \right),
    \label{eq:input-sld}
\end{equation}
Using these SLDs, the input SLD quantum Fisher matrix is
\begin{equation}
    (\mathcal F_Q)_{ij}
    =
    \frac{1}{2}
    \Tr[
    \rho_\theta
    (L_i^QL_j^Q+L_j^QL_i^Q)
    ].
    \label{eq:input-qfi}
\end{equation}
For observables $\bm O=(O_1,\ldots,O_m)$, write their expectation values as $\bm o(\theta)=(\Tr[O_1\rho_\theta],\ldots,\Tr[O_m\rho_\theta])$.
The symbol $D_{\bm O}$ denotes the $m\times d$ Jacobian, with entries $(D_{\bm O})_{\alpha i}=\partial_i o_\alpha(\theta)$, that maps a parameter displacement to the corresponding expectation-value change.
For any locally unbiased estimator $\hat{\bm o}(Y)$ from the record, the Cram\'er-Rao inequality gives
\begin{equation}
    {\rm Cov}_\theta(\hat{\bm o})
    \succeq
    \mathsf E_Y(\bm O),
    \qquad
    \mathsf E_Y(\bm O)
    :=
    D_{\bm O}J_Y^+D_{\bm O}^{T}.
    \label{eq:record-error}
\end{equation}
Here $+$ denotes the Moore--Penrose pseudoinverse on the estimable support.
Thus $\mathsf E_Y$ is the record-only lower-bound matrix for estimating the desired expectation values.
The corresponding input quantum lower-bound matrix is
\begin{equation}
    \mathsf E_Q(\bm O)
    :=
    D_{\bm O}\mathcal F_Q^+D_{\bm O}^{T}.
    \label{eq:quantum-error}
\end{equation}
The record is obtained from the input model by a parameter-independent measurement channel, so SLD quantum Fisher information is monotone under data processing \cite{BraunsteinCaves1994,PetzSudar1996}.
Consequently $J_Y\preceq\mathcal F_Q$ and, on the common estimable support,
\begin{equation}
    \mathsf E_Y(\bm O)
    \succeq
    \mathsf E_Q(\bm O).
    \label{eq:record-qbound}
\end{equation}
WSU identify the excess
\begin{equation}
    \DeltaW(\bm O)
    :=
    \mathsf E_Y(\bm O)-\mathsf E_Q(\bm O)
    \succeq0
    \label{eq:wsu-error-matrix}
\end{equation}
as the estimation-theoretic measurement-error matrix, where $\mathsf E_Y$ and $\mathsf E_Q$ are the lower-bound matrices in Eqs.~(\ref{eq:record-error}) and (\ref{eq:quantum-error}), built from this $D_{\bm O}$.
For $\bm O=(A,B)$, write $D_{AB}:=D_{\bm O}$ and let $e_A=(1,0)^T$ and $e_B=(0,1)^T$ in the observable plane.
The WSU scalar errors are the diagonal projections
\begin{equation}
    \varepsilon_Y(A)
    :=
    e_A^T\DeltaW^{AB}e_A,
    \qquad
    \varepsilon_Y(B)
    :=
    e_B^T\DeltaW^{AB}e_B .
    \label{eq:wsu-scalar-errors}
\end{equation}
The left-hand side of Eq.~(\ref{eq:wsu-determinant-step}) is the product of these two diagonal projections of $\DeltaW^{AB}$.
The commutator bound comes from a stronger complex matrix inequality used in the WSU proof.
Let $\mathcal F_{\rm RLD}$ denote the right-logarithmic-derivative (RLD) Fisher matrix, whose inverse contains the nonsymmetrized quantum correlations.
On the estimable support, the RLD form of the quantum Cram\'er-Rao inequality gives the Hermitian positivity
\begin{equation}
    \mathsf K_{\rm WSU}^{AB}
    :=
    D_{AB}
    \left(
    J_Y^{-1}
    -
    \mathcal F_{\rm RLD}^{-1}
    \right)
    D_{AB}^{T}
    \succeq0 .
    \label{eq:wsu-rld-positive}
\end{equation}
In the observable basis $(A,B)$, the diagonal entries are
$(\mathsf K_{\rm WSU}^{AB})_{AA}=\varepsilon_Y(A)$ and
$(\mathsf K_{\rm WSU}^{AB})_{BB}=\varepsilon_Y(B)$.
Writing the off-diagonal element as $\kappa_{AB}$, its imaginary part is fixed by the commutator,
$|{\rm Im}\,\kappa_{AB}|=|\langle[A,B]\rangle_{\rho_\theta}|/2$.
The determinant condition $\det\mathsf K_{\rm WSU}^{AB}\ge0$ gives
$\varepsilon_Y(A)\varepsilon_Y(B)\ge|\kappa_{AB}|^2$.
Therefore
\begin{equation}
    \varepsilon_Y(A)\varepsilon_Y(B)
    \ge
    |\kappa_{AB}|^2
    \ge
    \frac{1}{4}
    |\langle[A,B]\rangle_{\rho_\theta}|^2 .
    \label{eq:wsu-determinant-step}
\end{equation}
This is the scalar WSU measurement-error uncertainty relation: the commutator term is not obtained from the projection definition alone, but from the RLD-based complex positivity in Eq.~(\ref{eq:wsu-rld-positive}).
Below we resolve each scalar error into instrument and record-restriction contributions.

{\it Instruments and the record-residual split.---}
We now refine the record-only description to a quantum instrument.
A measurement with an accessible outcome $y$ and a residual output system $R$ is described by a family of completely positive maps \cite{DaviesLewis1970,Ozawa1984}.
In its most familiar Kraus form, the outcome-$y$ branch is
\begin{equation}
    \rho_\theta
    \longmapsto
    \sum_\mu
    M_{y\mu}\rho_\theta M_{y\mu}^\dagger .
    \label{eq:kraus-branch}
\end{equation}
The operators $M_{y\mu}$ include the measurement interaction, the readout outcome, and any unobserved microscopic channel label $\mu$.
Trace preservation of the total measurement means $\sum_{y,\mu}M_{y\mu}^\dagger M_{y\mu}=I $.
More generally, we write the outcome branch as a completely positive map
\begin{equation}
    \M_y:\rho_\theta\mapsto \M_y(\rho_\theta)
\end{equation}
from input states to subnormalized states on $R$.
Equation (\ref{eq:kraus-branch}) is the special case
\begin{equation}
    \M_y(\rho_\theta)
    =
    \sum_\mu
    M_{y\mu}\rho_\theta M_{y\mu}^\dagger .
    \label{eq:instrument-kraus}
\end{equation}
The normalization condition is $\sum_y\Tr[\M_y(\rho)]=1$ for all normalized $\rho$.
The probability of observing $y$ is
\begin{equation}
    p_\theta(y)=\Tr[\M_y(\rho_\theta)],
    \label{eq:prob}
\end{equation}
and, for $p_\theta(y)>0$, the conditional residual state is
\begin{equation}
    \eta_{\theta,y}^R
    =
    \frac{\M_y(\rho_\theta)}{p_\theta(y)}.
    \label{eq:conditional}
\end{equation}
If the outcome is stored in an orthonormal classical register $Y$, the full output is the classical-quantum state
\begin{equation}
    \G_\theta^{YR}
    =
    \sum_y
    |y\rangle\!\langle y|
    \otimes
    \M_y(\rho_\theta)
    =
    \sum_y
    p_\theta(y)|y\rangle\!\langle y|
    \otimes
    \eta_{\theta,y}^{R}.
    \label{eq:cq-state}
\end{equation}
The marginal state of the record is
\begin{equation}
    \G_\theta^{Y}
    =
    \sum_y p_\theta(y)|y\rangle\!\langle y|.
\end{equation}
Its SLD Fisher information is exactly the classical record matrix $J_Y$ in Eq.~(\ref{eq:record-fi}).
A record-restricted estimator can use only $\G_\theta^Y$, equivalently the sampled record $Y$.
A residual-aware estimator may also measure or coherently process the conditional system $R$.

{\it Fisher information of an instrument output.---}
The record Fisher matrix $J_Y$ captures only the classical register.
The full classical-quantum output $\G_\theta^{YR}$ can contain more information, because the conditional state $\eta_{\theta,y}^R$ may still depend on the input parameter.
Let $L_{i,y}^R$ be the SLD for the conditional residual state,
\begin{equation}
    \partial_i\eta_{\theta,y}^R
    =
    \frac{1}{2}
    \left(
    L_{i,y}^R\eta_{\theta,y}^R
    +
    \eta_{\theta,y}^RL_{i,y}^R
    \right).
\end{equation}
The SLD of $\G_\theta^{YR}$ is block diagonal in $y$:
\begin{equation}
    L_i^{YR}
    =
    \sum_y
    |y\rangle\!\langle y|
    \otimes
    \left[
    s_i(y)I_R+L_{i,y}^R
    \right].
    \label{eq:sld-block}
\end{equation}
Using $\Tr[\eta_{\theta,y}^RL_{i,y}^R]=0$, the cross terms vanish and the SLD Fisher matrix decomposes as
\begin{equation}
    \mathcal F_{YR}
    =
    J_Y+\Fr,
    \label{eq:qfi-split}
\end{equation}
where
\begin{equation}
    (\Fr)_{ij}
    =
    \sum_y p_\theta(y)
    \frac{1}{2}
    \Tr[
    \eta_{\theta,y}^R
    (L_{i,y}^RL_{j,y}^R+L_{j,y}^RL_{i,y}^R)
    ].
    \label{eq:hidden-residual-fi}
\end{equation}
The second term is the hidden residual Fisher information; Eq.~(\ref{eq:qfi-split}) is the classical--quantum identity of Ref.~\cite{CombesFerrieJiangCaves2014}.
It is not in the classical record, but it is present in the conditional quantum output.
Because $\G_\theta^Y$ is obtained by discarding $R$, and $\G_\theta^{YR}$ is obtained from the input by a parameter-independent channel, Fisher monotonicity gives
\begin{equation}
    J_Y
    \preceq
    \mathcal F_{YR}
    =
    J_Y+\Fr
    \preceq
    \mathcal F_Q .
    \label{eq:fisher-hierarchy}
\end{equation}

{\it Residual-aware error matrices.---}
If the estimator can use the full classical-quantum output $YR$, the corresponding SLD Cram\'er-Rao lower-bound matrix is
\begin{equation}
    \mathsf E_{YR}(\bm O)
    :=
    D_{\bm O}
    (J_Y+\Fr)^+
    D_{\bm O}^{T}.
    \label{eq:residual-aware-error}
\end{equation}
Thus $\mathsf E_Y$ is the error lower bound for the record interface, while $\mathsf E_{YR}$ is the lower bound for a residual-aware quantum interface.

\begin{theorem}[Instrument-resolved WSU error]
On the common estimable support of $J_Y$, $J_Y+\Fr$, and $\mathcal F_Q$,
\begin{equation}
    \mathsf E_Q(\bm O)
    \preceq
    \mathsf E_{YR}(\bm O)
    \preceq
    \mathsf E_Y(\bm O)
    \label{eq:error-hierarchy}
\end{equation}
and the record-only lower-bound matrix decomposes as
\begin{equation}
    \mathsf E_Y(\bm O)
    =
    \mathsf E_{YR}(\bm O)
    +
    \PiR(\bm O),
    \qquad
    \PiR(\bm O)\succeq0 .
    \label{eq:error-split}
\end{equation}
If $J_Y$ is nonsingular and $\Fr=LL^T$, then
\begin{align}
    \PiR(\bm O)
    &=
    D_{\bm O}J_Y^{-1}L
    (I+L^TJ_Y^{-1}L)^{-1}
    L^TJ_Y^{-1}D_{\bm O}^{T}.
    \label{eq:schur-correction}
\end{align}
Consequently the WSU measurement-error matrix decomposes as
\begin{equation}
    \DeltaW(\bm O)
    =
    \Dinst(\bm O)
    +
    \Drec(\bm O),
    \label{eq:delta-decomposition}
\end{equation}
where
\begin{equation}
    \Dinst(\bm O)
    :=
    \mathsf E_{YR}(\bm O)-\mathsf E_Q(\bm O),
    \qquad
    \Drec(\bm O):=\PiR(\bm O).
    \label{eq:delta-definitions}
\end{equation}
Both terms are positive semidefinite.
\end{theorem}

\textit{Proof.---}
For clarity take $J_Y>0$; the singular case follows by restricting to the common support.
The Woodbury identity gives
\begin{align}
    &(J_Y+LL^T)^{-1}
    \nonumber\\
    &\quad
    =
    J_Y^{-1}
    -
    J_Y^{-1}L
    (I+L^TJ_Y^{-1}L)^{-1}
    L^TJ_Y^{-1}.
    \label{eq:woodbury-step}
\end{align}
Subtracting this expression from $J_Y^{-1}$ yields
\begin{equation}
    J_Y^{-1}-(J_Y+\Fr)^{-1}
    =
    J_Y^{-1}L
    (I+L^TJ_Y^{-1}L)^{-1}
    L^TJ_Y^{-1}.
    \label{eq:inverse-difference}
\end{equation}
The middle factor in Eq.~(\ref{eq:inverse-difference}) is the Schur complement $I-L^T(J_Y+\Fr)^{-1}L$ of $J_Y+\Fr$.
Multiplying Eq.~(\ref{eq:inverse-difference}) by $D_{\bm O}$ and $D_{\bm O}^T$ gives Eq.~(\ref{eq:schur-correction}).
The middle matrix is positive, hence $\PiR(\bm O)\succeq0$.
The full hierarchy in Eq.~(\ref{eq:error-hierarchy}) follows from the Fisher hierarchy in Eq.~(\ref{eq:fisher-hierarchy}) and inversion on the common support.
Equation~(\ref{eq:delta-decomposition}) is then obtained by subtracting $\mathsf E_Q$ from Eq.~(\ref{eq:error-split}).
Thus WSU's record-only measurement error is not a single loss: it contains intrinsic instrument loss plus record-restriction loss.

{\it Measurement uncertainty.---}
We now translate the error-matrix identity into an uncertainty relation.
It is useful first to recall the matrix structure already present in the standard uncertainty relation.
For centered observables $\Delta A=A-\langle A\rangle$ and $\Delta B=B-\langle B\rangle$, define
\begin{equation}
    G_\rho^{AB}
    =
    \begin{pmatrix}
    \langle(\Delta A)^2\rangle_\rho
    &
    \langle\Delta A\Delta B\rangle_\rho\\
    \langle\Delta B\Delta A\rangle_\rho
    &
    \langle(\Delta B)^2\rangle_\rho
    \end{pmatrix}.
    \label{eq:preparation-gram}
\end{equation}
This Gram matrix is positive because
\begin{equation}
    \left\langle
    (c_A\Delta A+c_B\Delta B)^\dagger
    (c_A\Delta A+c_B\Delta B)
    \right\rangle_\rho
    \ge0
    \label{eq:gram-positive}
\end{equation}
for arbitrary complex $c_A,c_B$.
Its off-diagonal element separates into a symmetric covariance and an antisymmetric commutator term,
\begin{equation}
    \langle\Delta A\Delta B\rangle_\rho
    =
    {\rm Cov}_\rho(A,B)
    +
    \frac{1}{2}\langle[A,B]\rangle_\rho ,
    \label{eq:cov-comm-split}
\end{equation}
where ${\rm Cov}_\rho(A,B)
    =
    \langle \Delta A\Delta B+\Delta B\Delta A\rangle_\rho/2$.
For Hermitian $A$ and $B$, ${\rm Cov}_\rho(A,B)$ is real, whereas $\langle[A,B]\rangle_\rho$ is purely imaginary.
Thus the determinant condition $\det G_\rho^{AB}\ge0$ gives
\begin{align}
    (\Delta A)^2_\rho(\Delta B)^2_\rho
    -
    {\rm Cov}_\rho(A,B)^2
    -
    \frac{1}{4}|\langle[A,B]\rangle_\rho|^2 \ge 0.
    \label{eq:gram-det-step}
\end{align}
The Robertson--Schr\"odinger relation is the determinant bound
\begin{equation}
    \det V_\rho^{AB}
    =
    (\Delta A)^2_\rho(\Delta B)^2_\rho
    -
    {\rm Cov}_\rho(A,B)^2
    \ge
    \frac{1}{4}
    |\langle[A,B]\rangle_\rho|^2 ,
    \label{eq:robertson-schrodinger}
\end{equation}
where the real covariance matrix is
\begin{equation}
    V_\rho^{AB}
    =
    \begin{pmatrix}
    (\Delta A)^2_\rho & {\rm Cov}_\rho(A,B)\\
    {\rm Cov}_\rho(A,B) & (\Delta B)^2_\rho
    \end{pmatrix}.
    \label{eq:preparation-covariance}
\end{equation}
For a canonical pair $[Q,P]=i$, we obtain $\det V_\rho^{QP}\ge1/4$.
The frequently quoted product inequality follows by dropping the nonnegative covariance contribution.

The determinant is the coordinate-independent error area, while the product of standard deviations is an axiswise projection.
Our measurement-error relation has the same matrix character, but the covariance matrix is replaced by a Cram\'er-Rao error matrix determined by the accessible interface.
For this observable pair, denote the resulting $2\times2$ error block by the superscript $AB$.
Define the input quantum error area by
\begin{equation}
    B_{AB}(\theta)
    :=
    \sqrt{\det\mathsf E_Q^{AB}}.
    \label{eq:bab-definition}
\end{equation}
In canonical minimum-uncertainty cases this reduces to
$B_{AB}=|\langle[A,B]\rangle|/2$, but in general $B_{AB}$ is the determinant of the quantum Cram\'er-Rao lower-bound block.
\begin{corollary}[Error-area hierarchy]
The hierarchy in Theorem~1 immediately gives the matrix uncertainty hierarchy
\begin{equation}
    \det\mathsf E_Y^{AB}
    =
    \det\!\left[
    \mathsf E_{YR}^{AB}
    +
    \Pi_R^{AB}
    \right]
    \ge
    \det\mathsf E_{YR}^{AB}
    \ge
    B_{AB}^2.
    \label{eq:error-area-hierarchy}
\end{equation}
\end{corollary}
This is the matrix version of the usual uncertainty-area statement, now resolved by the available measurement interface.

\begin{corollary}[Scalar WSU correction]
Define the scalar projections
$\varepsilon_{\rm inst}(A)=e_A^T\Dinst^{AB}e_A$ and
$\pi_R(A)=e_A^T\Drec^{AB}e_A$, and similarly for $B$.
Then the WSU scalar errors in Eq.~(\ref{eq:wsu-scalar-errors}) decompose as
\begin{align}
    \varepsilon_Y(A)
    =
    \varepsilon_{\rm inst}(A)+\pi_R(A),
    \nonumber\\
    \varepsilon_Y(B)
    =
    \varepsilon_{\rm inst}(B)+\pi_R(B),
    \label{eq:scalar-delta-split}
\end{align}
Since $\Drec^{AB}=\Pi_R^{AB}$, Eq.~(\ref{eq:schur-correction}) makes these residual penalties directly computable:
$\pi_R(A)=e_A^TD_{AB}[J_Y^{-1}-(J_Y+\Fr)^{-1}]D_{AB}^{T}e_A$, with $\pi_R(B)$ obtained by replacing $e_A$ with $e_B$.
Since $\Dinst^{AB}\succeq0$, the same split also yields a residual-improved scalar lower bound.
The record-only product obeys
\begin{equation}
    \varepsilon_Y(A)\varepsilon_Y(B)
    \ge
    \max\{C_{AB},R_{AB}\}.
    \label{eq:corrected-wsu-product}
\end{equation}
\end{corollary}
Here $C_{AB}=|\langle[A,B]\rangle_{\rho_\theta}|^2/4$ and
$R_{AB}=[\sqrt{\pi_R(A)\pi_R(B)}+\sqrt{\det\Dinst^{AB}}]^2$.
The second bound follows by expanding the left-hand side.
For a positive $2\times2$ matrix, the product of the two diagonal entries is at least the determinant; hence
\begin{equation}
    \varepsilon_{\rm inst}(A)\varepsilon_{\rm inst}(B)
    \ge
    \det\Dinst^{AB}.
    \label{eq:inst-det-bound}
\end{equation}
The cross terms satisfy the arithmetic--geometric mean inequality $u+v\ge2\sqrt{uv}$, applied to
$u=\varepsilon_{\rm inst}(A)\pi_R(B)$ and
$v=\varepsilon_{\rm inst}(B)\pi_R(A)$.
These two facts give $R_{AB}$.
The retained-residual excess product
$\varepsilon_{\rm inst}(A)\varepsilon_{\rm inst}(B)$ need not exceed $C_{AB}$; Eq.~(\ref{eq:corrected-wsu-product}) is a record-only bound.

As a physical example, consider single-mode displacement sensing.
Let $Q$ and $P$ be quadratures with $[Q,P]=i$, and take a coherent input for which $(\Delta Q)^2=(\Delta P)^2=1/2$ and ${\rm Cov}(Q,P)=0$.
A phase-space displacement is
\begin{equation}
    \rho_{q,p}
    =
    e^{-iqP+ipQ}\rho_0 e^{iqP-ipQ}.
    \label{eq:displaced-state}
\end{equation}
For an input centered at the origin, $q=\langle Q\rangle_{\rho_{q,p}}$ and $p=\langle P\rangle_{\rho_{q,p}}$.
For a pure-state unitary model generated by $G$, the quantum Fisher information is $4(\Delta G)^2$, where $(\Delta G)^2:=\langle G^2\rangle-\langle G\rangle^2$.
The generators of $q$ and $p$ are $P$ and $-Q$, so $\mathcal F_Q(q,p)=2I_2$.
A beam splitter of reflectivity $0<\kappa\le1$ sends the reflected port to a heterodyne detector, producing $Y=(Y_Q,Y_P)$, while the transmitted port is kept as $R$.
With heterodyne noise normalized to unit covariance,
$Y\sim
    \mathcal N\!\left(\sqrt{\kappa}(q,p),I_2\right),
$ and hence $J_Y=\kappa I_2$.
The transmitted mode is displaced by $\sqrt{1-\kappa}(q,p)$, giving $\Fr=2(1-\kappa)I_2$.
Taking $\bm O=(Q,P)$ and $D_{\bm O}=I_2$, the record-only and residual-aware error matrices are
\begin{equation}
    \mathsf E_Y^{QP}
    =
    \frac{1}{\kappa}I_2,
    \qquad
    \mathsf E_{YR}^{QP}
    =
    \frac{1}{2-\kappa}I_2.
    \label{eq:optical-errors}
\end{equation}
The positive correction is therefore
\begin{equation}
    \Pi_R^{QP}
    =
    \frac{2(1-\kappa)}
    {\kappa(2-\kappa)}
    I_2 .
    \label{eq:optical-pi}
\end{equation}
Projecting this matrix onto the two quadrature axes gives $\pi_R(Q)=\pi_R(P)=2(1-\kappa)/[\kappa(2-\kappa)]$.
The input quantum lower-bound matrix is $\mathsf E_Q^{QP}=I_2/2$.
Therefore
\begin{equation}
    \Dinst^{QP}
    =
    \frac{\kappa}{2(2-\kappa)}I_2,
    \qquad
    \Drec^{QP}
    =
	    \frac{2(1-\kappa)}{\kappa(2-\kappa)}I_2 .
    \label{eq:optical-deltas}
\end{equation}
Thus $\varepsilon_{\rm inst}(Q)\varepsilon_{\rm inst}(P)=[\kappa/(2(2-\kappa))]^2<C_{QP}=1/4$ for $0<\kappa<1$.
By contrast, the record-only WSU excess-error matrix is
\begin{equation}
    \DeltaW^{QP}
    =
    \frac{2-\kappa}{2\kappa}I_2
    =
    \Dinst^{QP}
    +
    \Drec^{QP}.
    \label{eq:optical-wsu-error}
\end{equation}
The scalar WSU product is therefore
\begin{equation}
    \varepsilon_Y(Q)\varepsilon_Y(P)
    =
    R_{QP}
    =
    \left(\frac{2-\kappa}{2\kappa}\right)^2
    \ge
    C_{QP}
    =
    \frac{1}{4},
    \label{eq:optical-wsu-product}
\end{equation}
Thus the residual-improved bound in Corollary~2 is saturated, while equality with the commutator bound occurs at $\kappa=1$.
For $0<\kappa<1$, the extra WSU error is record-restriction loss; the missing precision remains in the transmitted mode.

{\it Discussion.---}
The decomposition is a design rule: $\Dinst$ is irreversible instrument loss, whereas $\Drec=\PiR$ is precision left in an unread residual system.
It separates record-only sensing, syndrome decoding, thermodynamic-trajectory estimates, and sequential error--disturbance settings, where the disturbed post-measurement state is a residual output \cite{Watanabe2014}.
Instrument-level information also appears in information erasure and quantum imprint \cite{SuzukiItoOhzeki2026} and is also relevant to weak-value protocols.
The same access hierarchy applies to the multiparameter Holevo--Cram\'er--Rao cost \cite{Holevo2011}, although incompatible observables may prevent simultaneous use of the residual information.

\begin{acknowledgments}
We received financial support from the Cross-ministerial Strategic Innovation Promotion Program (SIP) of the Cabinet Office (No. 23836436).
\end{acknowledgments}

\bibliographystyle{apsrev4-2}
\bibliography{references}

\end{document}